\documentclass[conference]{IEEEtran}
\IEEEoverridecommandlockouts
\usepackage{cite}
\usepackage{amsmath,amssymb,amsfonts}
\usepackage{algorithmic}
\usepackage{graphicx}
\usepackage{textcomp}
\usepackage{xcolor}
\usepackage[hyphens]{url}
\usepackage[hidelinks]{hyperref} 
\usepackage{mathtools} %
\usepackage{amsthm}
\usepackage{algorithm}
\usepackage{multirow}
\usepackage{makecell}
\usepackage[table]{xcolor}

\newtheorem{proposition}{Proposition}

\DeclareMathOperator*{\argmax}{arg\,max}

\DeclareMathOperator{\ext}{ext}

\allowdisplaybreaks

\def\BibTeX{{\rm B\kern-.05em{\sc i\kern-.025em b}\kern-.08em
    T\kern-.1667em\lower.7ex\hbox{E}\kern-.125emX}}

\begin{document}

\title{DER Day-Ahead Offering: A Neural Network Column-and-Constraint Generation Approach}

\author{
\IEEEauthorblockN{Weiqi Meng, Hongyi Li, Bai Cui}
\IEEEauthorblockA{
Department of Electrical and Computer Engineering\\
Iowa State University, Ames, IA, USA\\
Email: \texttt{mengwq@iastate.edu, hongyili@iastate.edu, baicui@iastate.edu}
}
\thanks{The work was supported in part by the U.S. Department of Energy Solar Energy Technologies Office under Grant DE-EE0011374; and in part by the National Science Foundation under Grant 2523935.}
}


\maketitle

\IEEEpubid{
\begin{minipage}{\textwidth}
\centering \footnotesize \vspace{1in}
\copyright~2026 IEEE. Personal use of this material is permitted. Permission from IEEE must be obtained for all other uses, including reprinting/republishing this material for advertising or promotional purposes, creating new collective works for resale or redistribution to servers or lists, or reuse of any copyrighted components of this work in other works.
\end{minipage}
}

\begin{abstract}
In the day-ahead energy market, the offering strategy of distributed energy resource (DER) aggregators must be submitted before the uncertainty realization in the form of price-quantity pairs. This work addresses the day-ahead offering problem through a two-stage adaptive robust stochastic optimization model, wherein the first-stage price-quantity pairs and second-stage operational commitment decisions are made before and after DER uncertainty is realized, respectively. Uncertainty in day-ahead price is addressed using a stochastic programming-based approach, while uncertainty of DER generation is handled through robust optimization. To address the max-min structure of the second-stage problem, a neural network-accelerated column-and-constraint generation method is developed. A dedicated neural network is trained to approximate the value function, while optimality is maintained by the design of the network architecture. Numerical studies indicate that the proposed method yields high-quality solutions and is up to $100$ times faster than Gurobi and $33$ times faster than classical column-and-constraint generation on the same 1028-node synthetic distribution network.

\end{abstract}


\section{Introduction}
With the increasing penetration of distributed energy resources (DERs), DER aggregators face new challenges in participating in day-ahead electricity markets with significant generation and price uncertainty \cite{iea_unlocking_2022}. An offering strategy defines how an aggregator determines the quantity and price of energy or flexibility to be offered in the market while satisfying the operational and network constraints. Significant attention has been drawn to the design of economic day-ahead offering strategies through optimization-based methods, including the use of robust optimization or stochastic programming \cite{kaya_fifty_2026}.

Currently, the day-ahead DER offering problem is typically formulated as a two-stage adaptive robust optimization to capture price and generation uncertainty: the aggregator determines the ``here-and-now'' decisions (offering price-quantity pairs) before the realization of price and DER uncertainty and optimizes the ``wait-and-see'' recourse decisions (DER dispatch commands) after the uncertainty is revealed. However, robust optimization can be overly conservative in certain settings \cite{baringo_stochastic_2017}. By contrast, stochastic programming-based approaches model uncertainty through an explicit set of sampled scenarios. Its disadvantages lie in the requirement for knowledge of the explicit probabilistic distribution of the uncertain parameters, which might lead to the so-called ``curse of dimensionality'' \cite{plazas_multimarket_2005}. To balance both approaches effectively, the application of two-stage adaptive robust stochastic optimization (2S-ARSO) becomes attractive. With the advancement of optimization methods, existing approaches to solve 2S-ARSO include: 1) approximation algorithms \cite{bertsimas_theory_2010}; 2) Benders decomposition \cite{bertsimas_adaptive_2013}; 3) column-and-constraint generation (CCG) \cite{zeng_solving_2013}; and 4) scenario reduction method \cite{sun_robust_2021}. However, these methods may not scale well for large-scale problems. Given the computational burden stemming from the max-min recourse structure and the large scenario sets inherent to 2S-ARSO, its practical deployment is often constrained, particularly in industrial environments where operational time and computational resources are limited. 

In recent years, growing interest has focused on embedding neural network (NN) models within classical algorithmic pipelines to improve efficiency and scalability in practice \cite{zhang_survey_2023}. Much of the machine learning (ML) literature for linear programming (LP) targets deterministic settings, making it inapplicable to our 2S-ARSO problem. The core difficulty lies in its intrinsic infinite-dimensional characteristic: the recourse decisions are decision rules defined over the uncertainty space and parameterized by the here-and-now decisions, yielding an infinite-dimensional functional dependence. Second, 2S-ARSO embeds a special tri-level structure: choosing a here-and-now decision, identifying the worst-case uncertainty realization, and optimizing the ensuing recourse decision. Any practical method must address all three components in a unified manner. Third, the computational burden of solving the 2S-ARSO is negligible when considering network constraints and a large number of uncertainty scenarios. Finally, during the iteration of CCG, the induced optimization becomes progressively harder and may become computationally intractable as the number of iterations grows.

To address these challenges, we develop a tailored NN-accelerated CCG framework that relieves the computational burden of 2S-ARSO for DER day-ahead offering strategy design. The NN-based CCG algorithm is based on \cite{dumouchelle_neur2ro_2024}, while we adapt it to our specific problem setting. Specifically, an NN is trained to estimate the optimal cost of the value function of the recourse problem and to identify the corresponding worst-case uncertainty realization. The resulting predictions eliminate the need to solve the infinite-dimensional linear recourse problem and streamline the classical CCG procedure (see Section \ref{ML-CCG}). To the best of our knowledge, this is the first study leveraging ML to address the 2S-ARSO-based DER day-ahead offering problem. The training process is conducted once, and the trained NNs can be easily utilized for day-ahead offering decision making while achieving significant speed-ups that outperform the state-of-the-art commercial solvers. Our key contributions are summarized as follows:
\begin{enumerate}
\item We formulate the DER day-ahead offering problem as a hybrid 2S-ARSO model with network constraints and a polyhedral DER uncertainty set.
\item We develop an NN-accelerated CCG algorithm that tightens recourse constraints and yields an order-of-magnitude speedup for 2S-ARSO.
\item We design a tailored problem dimension-invariant neural architecture that jointly estimates worst-case realizations and the recourse value.
\item We demonstrate up to $100$-fold speedup of the NN-accelerated CCG over conventional CCG on synthetic networks using historical price data.
\end{enumerate}

The rest of the paper is organized as follows. Section~\ref{2S-ARSO} introduces the compact formulation of the day-ahead offering problem as a 2S-ARSO. Section \ref{2S-ARSO:Solution} proposes the NN-accelerated CCG solution algorithm. Numerical studies are included in Section \ref{NumericalStudy}. Section \ref{Conclusion} concludes the paper.

\section{Hybrid Adaptive Robust Stochastic Optimization Formulation}
\label{2S-ARSO}
In this work, we consider the day-ahead offering problem for a DER aggregator over 24 hours. The aggregator submits price-quantity offers (first-stage, ``here-and-now'') before DER uncertainty is realized.
After DER uncertainties are revealed, the aggregator makes operational recourse decisions (second-stage, ``wait-and-see'') regarding DER dispatch, which are constrained by device and power flow constraints.

We model the 2S-ARSO day-ahead offering problem as a 2S-ARO problem under sampled price trajectories. For the 2S-ARO model, the decision vector is split into ``here-and-now'' variables ($\boldsymbol{x}$) and ``wait-and-see'' (recourse) variables ($\boldsymbol{y}$). The here-and-now decisions (submitted price-quantity pairs) are fixed before the DER uncertainty is realized, whereas the cost-minimizing recourse decisions (DER dispatch and other real-time decisions) are determined after the uncertainty is revealed. The day-ahead offering model is formulated as the following risk-neutral expectation over all sampled day-ahead price scenarios:
\begin{multline}
\label{eq:obj_offering}
\min_{\boldsymbol{x}\in \mathcal{X}} \; \mathbb{E} _{\omega \in \Omega}[f(\boldsymbol{x},\omega )]
= \\
\min_{\boldsymbol{x}\in \mathcal{X}} \sum_{\substack{\omega \in \Omega \\ \sum \rho_\omega = 1}} {\rho _{\omega}}\left\{ \boldsymbol{c}^{\top}(\omega )\boldsymbol{x}+ \max_{\boldsymbol{\xi }\in \mathcal{U}} \min_{\boldsymbol{y}\in \mathcal{Y} (\boldsymbol{x},\boldsymbol{\xi} )} \boldsymbol{b}^{\top}\boldsymbol{y} \right\}.
\end{multline}


In the problem above, the day-ahead price at the first stage is modeled by stochastic price trajectories $\boldsymbol{\omega} \in \Omega$ sampled through the Markov process (each trajectory is assigned a nonnegative weight $\rho _{\omega}$) \cite{zheng_arbitraging_2022}. The DER uncertainty is modeled through the budgeted uncertainty set $\boldsymbol{\xi} \in \mathcal{U}$ \cite{thiele_robust_2010}. The first-stage cost $\boldsymbol{c}^\top(\boldsymbol{\omega}) \boldsymbol{x}$ describes the negative value of the day-ahead offering revenue given the price uncertainty $\boldsymbol{\omega}$. The recourse cost $\boldsymbol{b}^{\top} \boldsymbol{y}$ represents the cost of real-time DER operation given the realization of DER uncertainty $\boldsymbol{\xi}$. The feasible regions of the first- and second-stage problems are denoted by $\mathcal{X}\subseteq \mathbb{R}^n$ and $\mathcal{Y}\subseteq \mathbb{R}^m$, respectively. To be consistent with U.S. electricity market regulations \cite{kim_benefits_2021}, our framework discourages inter-settlement arbitrage between two-settlement markets. This is done by modeling the real-time price as a small deviation from the day-ahead one in such a way that penalizes any intentional real-time power deviation.

The feasible region of the first-stage problem $\mathcal{X}$ is a bounded polyhedron (polytope) given by
\begin{equation}
\label{first_stage_FR}
\mathcal{X} = \{ \boldsymbol{x} \in \mathbb{R}^n : 
\boldsymbol{A}\boldsymbol{x} \geq \boldsymbol{d}\},
\end{equation}
which contains both the monotonicity constraint of offering curves and the capacity constraints of offering quantities. The feasible region of the second-stage problem $\mathcal{Y}(\boldsymbol{x},\boldsymbol{\xi})$ is
\begin{equation}
\label{CompactRPFR}
    \mathcal{Y}(\boldsymbol{x},\boldsymbol{\xi}) =\left\{ \boldsymbol{y} \in \mathbb{R}^m: \boldsymbol{E}\boldsymbol{x}+\boldsymbol{Fy} \ge \boldsymbol{g}+\boldsymbol{H\xi } \right\}, 
\end{equation}
which includes DERs power availability constraints, DER physical constraints, and linear power flow constraints \cite{gan_convex_2014}.

\section{Neural Network-Accelerated Column-and-Constraint Generation Algorithm}
\label{2S-ARSO:Solution}

For each day-ahead price scenario, the 2S-ARSO day-ahead
offering model \eqref{eq:obj_offering} can be rewritten as: 
\begin{equation}
\label{CompactOB}
    \underset{\boldsymbol{x}\in \mathcal{X}}{\min} \; \boldsymbol{c}\left( \boldsymbol{\omega} \right) ^{\top}\boldsymbol{x}+\mathcal{Q} \left( \boldsymbol{x},\boldsymbol{\xi} \right),
\end{equation}
where $\mathcal{Q}(\cdot)$ denotes the value function of the 2S-ARSO and  is given by:
\begin{equation}
\label{CompactRP}
\mathcal{Q} \left( \boldsymbol{x}\right) =\underset{\boldsymbol{\xi }\in \mathcal{U}}{\max}\,\,\underset{\boldsymbol{y}\in \mathcal{Y} \left( \boldsymbol{x},\boldsymbol{\xi} \right)}{\min}\boldsymbol{b}^{\top}\boldsymbol{y}.
\end{equation}

The problem in \eqref{CompactRP} identifies the worst-case scenario $\boldsymbol{\xi}$ over the inner minimization problem. Given the polyhedral structure of $\mathcal{Y}$ and the right-hand side (RHS) uncertainty of $\boldsymbol{\xi}$, under the relatively complete recourse assumption, the inner minimization problem is a convex piecewise linear function in $\boldsymbol{\xi}$, whose maximization over a polytopic uncertainty set $\mathcal{U}$ is attained at one of its extreme points \cite[Sect. 3.2.2]{sun_robust_2021}.


\subsection{Column-and-Constraint Generation Algorithm}
\label{CCGalg}

The CCG algorithm solves problem \eqref{CompactOB} by exploiting the structural insights mentioned above in an iterative master problem/subproblem fashion. A detailed description of the algorithm can be found in \cite[Sect. 3.2.6]{sun_robust_2021}. We provide a very brief introduction here to facilitate further development. 

The CCG algorithm can be viewed as a constraint generation algorithm: for a given first-stage candidate solution $\boldsymbol{x}^{(k)}$, it identifies a worst-case scenario $\boldsymbol{\xi}^{(k)} \in \ext(\mathcal{U})$ for the second-stage problem and generates a set of new constraints for the master problem, where $\ext(\cdot)$ denote the set of extreme points of a polytope. At iteration $k$, the master problem of \eqref{CompactOB} is
\begin{subequations}
\label{CompactCCG}
\begin{align}
\min
\quad & \eta \\
\big( \text{over} \quad & {\boldsymbol{x}\in\mathcal{X}, \eta, \boldsymbol{y}^{(1)}, \ldots, \boldsymbol{y}^{(k)}} \big) \nonumber \\
\text{s.t.} \quad
& \eta \;\ge\; \boldsymbol{c}(\boldsymbol{\omega})^{\top}\boldsymbol{x} \;+\; \boldsymbol{b}^{\top}\boldsymbol{y}^{(i)},
&& i = 1, \ldots, k \\
& \boldsymbol{E} \boldsymbol{x}+\boldsymbol{F} \boldsymbol{y}^{(i)} \ge \boldsymbol{H}\boldsymbol{\xi}^{(i)} + \boldsymbol{g}
&& i = 1, \ldots, k,
\end{align}
\end{subequations}
where $\boldsymbol{\xi}^{(i)} \in \ext(\mathcal{U})$. It is straightforward to see that problem \eqref{CompactCCG} is a relaxation of \eqref{CompactOB} and thus provides a lower bound. 

To certify the validity of the current first-stage solution $\boldsymbol{x}^{(k)}$ and to identify the worst-case violation, we formulate a subproblem that computes the worst-case uncertainty realization for $\boldsymbol{x}^{(k)}$. By invoking LP strong duality, the inner minimization in \eqref{CompactRP} admits the following equivalent reformulation:
\begin{equation}
\label{DualinnerRP}
    \mathcal{Q}(\boldsymbol{x})
    = \max_{\substack{j=1,\ldots, r \\ \ell = 1,\ldots, m}} (\boldsymbol{\pi}^{(j)})^\top (\boldsymbol{g} - \boldsymbol{E}\boldsymbol{x}) + (\boldsymbol{\pi}^{(j)})^\top\boldsymbol{H}\boldsymbol{\xi}^{(\ell)}
\end{equation}
where $\boldsymbol{\pi}^{(j)}$ are the extreme points of $\mathcal{P} := \{ \boldsymbol{\pi} \ge 0:  \boldsymbol{F}^\top \boldsymbol{\pi} = \boldsymbol{b} \}$. Since both $\mathcal{P}$ and $\mathcal{U}$ are convex polytopes, the worst-case value $\mathcal{Q}(\boldsymbol{x})$ is attained at one of their extreme points, which allows \eqref{DualinnerRP} to be represented as a finite maximization over $\ext(\mathcal{P})$ and $\ext(\mathcal{U})$. Given the current first-stage solution $\boldsymbol{x}^{(k)}$, the corresponding upper bound $\mathrm{UB}$ to \eqref{CompactOB} is updated as
\begin{equation}
    \mathrm{UB} = \min \left\{ \mathrm{UB}, \boldsymbol{c}(\omega)^\top \boldsymbol{x}^{(k)} + \mathcal{Q}(\boldsymbol{x}^{(k)}) \right\},
\end{equation}

When the mismatch between the upper and lower bounds is less than the preset convergence criterion, the CCG algorithm terminates and the optimal value and solution are obtained.

\subsection{NN-Accelerated Column-and-Constraint Generation}
\label{ML-CCG}

Although $\mathcal{Q} \left( \boldsymbol{x}\right)$ is a convex piecewise linear function of $\boldsymbol{x}$, evaluating its value for any fixed $\boldsymbol{x}$ is a hard problem due to the presence of bilinear terms in \eqref{DualinnerRP} which destroys convexity in $\boldsymbol{\pi}$ and $\boldsymbol{\xi}$. The computational burden of \eqref{DualinnerRP} becomes even worse with the consideration of power flow constraints, which is normally ignored \cite{baringo_stochastic_2017} but considered herein. We develop a neural estimator that, given a first-stage decision, predicts the optimal objective value of the corresponding second-stage problem with specific uncertainty realization with high accuracy. The NN is then integrated within CCG to accelerate the solution of both the master problem and the subproblem. It does so using one mixed-integer linear program (MILP)-representable NN with learnable parameters $\Theta$ that can assess, for a first-stage decision $\boldsymbol{x}$ under uncertain $\boldsymbol{\xi}$, the following:
\begin{equation}
\label{NNforRP}
\mathrm{NN}_{\Theta}\left( \boldsymbol{x},\boldsymbol{\xi} \right) \approx 
\min_{\boldsymbol{y} \in \mathcal{Y} \left( \boldsymbol{x},\boldsymbol{\xi } \right)} \boldsymbol{c}\left( \boldsymbol{\omega} \right) ^{\top}\boldsymbol{x}+\boldsymbol{b}^{\top}\boldsymbol{y} 
\end{equation}
where $\Theta$ denotes the set of trainable parameters of the network. The detailed architecture of $\mathrm{NN}_{\Theta}(\cdot)$ is illustrated in Fig.~\ref{fig1}. 

\subsubsection{Master Problem}
For the master problem \eqref{CompactCCG}, given a finite subset of uncertainty realizations $\{\boldsymbol{\xi}^{(1)}, \ldots, \boldsymbol{\xi}^{(k)}\}$, we reformulate the problem using an $\argmax$ estimator over the DER uncertainty scenarios. The estimator identifies the realization that maximizes the surrogate-predicted recourse value, which in turn serves as a proxy for the original second-stage objective value. Given a trained NN with parameters $\Theta^*$ that approximate the recourse cost, the surrogate master problem can be reformulated as:
\begin{subequations}
\label{NNforMP}
\begin{align}
    \min_{\boldsymbol{x}, \boldsymbol{y}} \quad & \boldsymbol{c}(\boldsymbol{\omega})^{\top} \boldsymbol{x} + \boldsymbol{b}^{\top}\boldsymbol{y} \\
    \text{s.t.} \quad & \boldsymbol{A}\boldsymbol{x} \geq \boldsymbol{d} \\
    & \boldsymbol{Fy} \ge \boldsymbol{H}\boldsymbol{\xi} + \boldsymbol{g} -\boldsymbol{E}\boldsymbol{x}  \\
    & \boldsymbol{\xi} \in \argmax_{\boldsymbol{\xi}^{(1)}, \ldots, \boldsymbol{\xi}^{(k)}} \left\{ \mathrm{NN}_{\Theta^*}\left( \boldsymbol{x},\boldsymbol{\xi} \right) \right\} 
\end{align}
\end{subequations}
Here, the outer approximation of the value function through Benders multicut in \eqref{CompactCCG} is replaced by a fixed-size surrogate NN model, which enables much more efficient computation. The $\argmax$ operator in \eqref{NNforMP} can be equivalently reformulated as an MILP using standard big-$M$ construction, where auxiliary binary variables and linear constraints are used to describe the maximization selection logic.


\subsubsection{Subproblem}
We replace the evaluation of the value function $\mathcal{Q}(\cdot)$ in \eqref{DualinnerRP} by its approximation $\mathrm{NN}_{\Theta^*}(\boldsymbol{x}^*,\boldsymbol{\xi})$ given the uncertainty realization. Given the solution of the master problem $\boldsymbol{x}^*$, the subproblem seeks to certify its optimality by evaluating 
the worst-case uncertainty realization that maximizes the  second-stage cost. We assume complete recourse through the payment of penalties for power mismatches so that no feasibility check is needed. We therefore solve
\begin{equation}
\label{eq:AP-NN-obj}
\boldsymbol{\xi}^*
\in
\operatorname*{arg\,max}_{\boldsymbol{\xi}\in\mathcal{U}}
\ \mathrm{NN}_{\Theta^*}(\boldsymbol{x}^*,\boldsymbol{\xi}),
\end{equation}
Since $\mathcal{U}$ is polytopic, \eqref{eq:AP-NN-obj} is an MILP that is solvable by off-the-shelf solvers.

To determine whether a new $\boldsymbol{\xi}^* \in \mathcal{U}$ has been identified, we compare
its prediction to those already in the master problem:
\begin{equation}
\label{eq:AP-stopping}
\mathrm{NN}_{\Theta^*}(\boldsymbol{x}^*,\boldsymbol{\xi}^*)
\;\ge\;
\max_{i = 1, \ldots, k}
\mathrm{NN}_{\Theta^*}(\boldsymbol{x}^*,\boldsymbol{\xi}^{(i)})
\;+\;\varepsilon,
\end{equation}
where $\varepsilon > 0$ is a small number. If \eqref{eq:AP-stopping} holds, we add $\boldsymbol{\xi}^*$ to the active scenario set in \eqref{NNforMP}, increment $k$ by $1$, and re-solve the master problem. The algorithm terminates otherwise.


As the iteration count rises, evaluations of the neural modules impose increasing computational overhead. We therefore adopt an architecture that can implement NN-accelerated CCG both efficiently and effectively. After sampling historical data, we embed $\boldsymbol{x}$ and $\boldsymbol{\xi}$ using dedicated low-dimensional embedding NNs $(\Phi_{\boldsymbol{x}} \text{ and } \Phi_{\boldsymbol{\xi}})$ for each instance. The resulting embeddings are concatenated and fed to the value network $\Phi_{\Theta^\ast}$, which estimates the worst-case scenario and the optimal recourse objective and serves as an oracle module within the pipeline for subsequent optimization. The neural architecture for solving the recourse problem is shown in Figure \ref{fig1}. For the master problem, since the scenario parameters are fixed, the scenario embeddings $\Phi_{\xi}(\xi^{(k)})$ can be precomputed via a forward pass for each scenario. In the finalized master problem, only the decision-embedding network $\Phi_{\boldsymbol{x}}$ and the value NN $\Phi_{\Theta^*}$ are required. Both can be implemented with standard ReLU-activated networks that lead to MILP formulations. In addition, duplicating these modules across scenarios remains computationally tractable. For the second-stage recourse problem, it suffices to encode only the decision embedding $\Phi_{\boldsymbol{x}}$ together with the value network $\Phi_{\Theta^\ast}$, since $\Phi_{\boldsymbol{x}}(\boldsymbol{x})$ can be precomputed efficiently via a single forward pass.

\begin{figure}[!t]
\centering
\includegraphics[width=3.4 in]{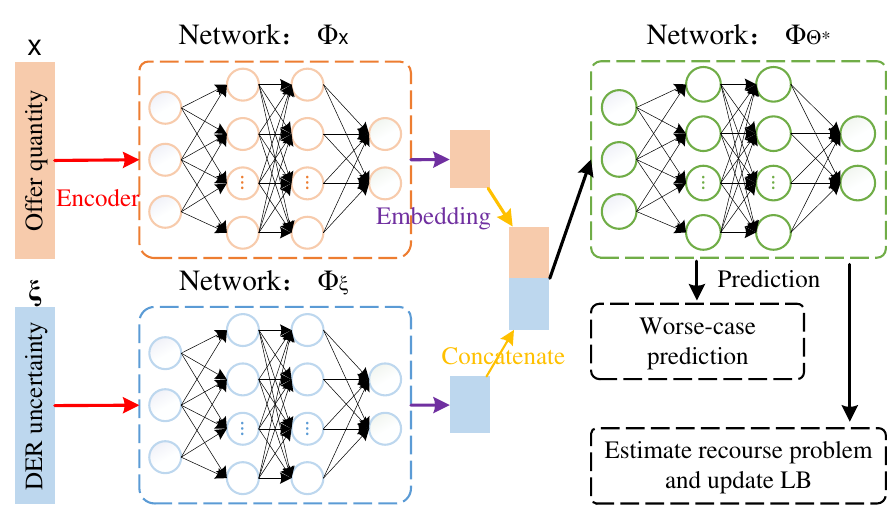} 
\caption{Neural architecture of NN-accelerated CCG}
\label{fig1}
\end{figure}

To generalize across instances, the architecture should be invariant to variable cardinality and ordering; and robust to permutations of constraint/objective coefficients. We enforce these properties by set-based NNs with shared parameters, which yield permutation-invariant encoders \cite{zaheer_deep_2018} and improve generalization to unseen instances. The detailed generalization process includes the following: Before feeding the first-stage solution and DER uncertainty scenarios to the embedding networks, these variables are decomposed into per-variable representations containing their values, constraint coefficients, and objective weights. To avoid underestimating the worst-case recourse cost, we check each NN-identified candidate scenario by solving the exact second-stage recourse problem. If it violates the tolerance, we add it to the active set and continue.
The resulting set of per-variable embeddings is then aggregated (e.g., sum pooling) and passed through a feedforward network to produce the first-stage and scenario representations. Due to space limitations, we omit further architectural details and refer to this component as the ``Encoder" in Fig.~\ref{fig1}, where the corresponding arrow is highlighted in red. The detailed description of the NN-accelerated CCG algorithm is summarized in Algorithm \ref{alg:neuraro-ccg-obj}. The finite-termination property of the algorithm is established in the Appendix.

\begin{algorithm}[t]
\caption{NN-Accelerated CCG Algorithm}
\label{alg:neuraro-ccg-obj}
\begin{algorithmic}[1]
\STATE \textbf{Input:}  uncertainty set $\mathcal{U}$;
objective NN $\Theta^*$; small tolerance $\varepsilon>0$.
\STATE \textbf{Output:} the optimal $\boldsymbol{x}^*\in\mathcal{X}$.
\STATE \textbf{Init:} pick any $\boldsymbol{\xi}_0\in\ext(\mathcal{U})$; $\mathcal{S}_0\leftarrow\{\boldsymbol{\xi}_0\}$; $k\leftarrow 0$;
\WHILE{not \texttt{converged}}
  \STATE Solve the master problem \eqref{NNforMP} with the active scenario set $\mathcal{S}^{(k)}$ to get $\boldsymbol{x}^{(k+1)}$.
  \STATE Derive the optimal solution of the subproblem
         $\boldsymbol{\xi}^* \in \argmax_{\boldsymbol{\xi}\in\mathcal{U}}\ \mathrm{NN}_{\Theta^*}(\boldsymbol{x}^{(k+1)},\boldsymbol{\xi})$.
  \IF{$\mathrm{NN}_{\Theta^*}(\boldsymbol{x}^{(k+1)},\boldsymbol{\xi}^*)\ge
       \max_{\boldsymbol{\xi}\in\mathcal{S}^{(k)}}\mathrm{NN}_{\Theta^*}(\boldsymbol{x}^{(k+1)},\boldsymbol{\xi})+\varepsilon$} \label{alg:ineq}
     \STATE $\mathcal{S}^{(k+1)}\leftarrow \mathcal{S}^{(k)}\cup\{\boldsymbol{\xi}^*\}$;
  \ELSE
     \STATE $\mathcal{S}^{(k+1)}\leftarrow \mathcal{S}^{(k)}$\quad \texttt{Converged}
  \ENDIF
  \STATE $k\leftarrow k+1$
\ENDWHILE
\STATE \textbf{return} $\boldsymbol{x}^{(k)}$
\end{algorithmic}
\end{algorithm}

\section{Numerical Study}
\label{NumericalStudy}
\subsection{Experiment Setup}
To validate the effectiveness of the NN-accelerated CCG algorithm, we test it on a 1028-node synthetic distribution network and is employed for comparison against state-of-the-art benchmark approaches. The synthetic distribution system is equipped with various resources in different busses, including rooftop solar PV (10-kW), behind-the-meter battery energy storage (Tesla’s Powerwall model, 11.3 kW/14.5 kWh), and an average load of 1.3 MW (peaking at 1.7 MW). All experiments are executed in Google Colab on a 64-bit Windows 11 workstation with an Intel® Core™ i7-6700HQ CPU running at 2.64 GHz and 24 GB of RAM, using GUROBI 12.0, and gurobi-machinelearning-1.5.5.dev0 is used to embed NNs into MILPs. The basic unit in this work is a fully-connected multilayer perceptron (MLP) with Rectified Linear Unit (ReLU) activation functions. We assess the performance of NN-accelerated CCG using two metrics: (i) optimality gap, computed as the mean absolute error (MSE) with the traditional CCG solution as the reference, and (ii) computational speedup relative to a classical CCG baseline. To test the performance of NN-accelerated CCG algorithm, we benchmark it with two state-of-the-art methods: 1) Gurobi \cite{gurobi}, which can directly solve the 2S-ARSO problem through a large MILP; 2) traditional CCG algorithm \cite{zeng_solving_2013}. 

\subsection{Structure and Accuracy of Neural Network}
For data training, we sample problem instances, candidate first-stage decisions, and DER uncertainty scenarios to construct feature vectors encoding constraint coefficients and objective weights. Labels are obtained by solving the second-stage recourse, which is a tractable deterministic MILP as $(\boldsymbol{x},\boldsymbol{\xi})$ are fixed. These data collection procedures are independent and parallelizable. The data collection times, training times, and total times are 562s, 844s, and 1,314s for the 2S-ARSO under 25 price trajectories, respectively. To keep the master and recourse problems tractable, we train small networks that achieve low MSE; therefore, we do not perform elaborate hyperparameter tuning, as further optimization would primarily refine the already strong numerical results. For the overall 2S-ARSO problem, we draw $1{,}000$ instances, generate $5$ first-stage decisions per instance, and sample $20$ scenarios per decision, producing $100{,}000$ labeled examples split into $80{,}000$ for training and $20{,}000$ for validation. Both embedding NNs employ two-layer MLPs with layer widths $[64,8]$. The value NN has one hidden layer of width $[8]$. We undergo 500 training epochs with validation MSE computed every 10 epochs. Model selection criteria prioritize minimal mean absolute validation error. Empirical findings reveal a close alignment between the MSE of the training and testing dataset, with the NN achieving sub-1.3\% MSE, indicating strong accuracy performance.

\subsection{Effectiveness of NN-Accelerated CCG}

A comparative evaluation of Gurobi, classical CCG, and the NN-accelerated CCG on the synthetic 1028-bus system, across varying numbers of sampled day-ahead price trajectories, is summarized in Table~\ref{tab:1028bus-perf}. As shown in Table \ref{tab:1028bus-perf}, compared with the Gurobi-based method, the proposed algorithm yields a reasonable optimal objective value with a gap of $0.1\%$ across all instances. This indicates that the trained ReLU-activated NN can accurately estimate the worst-case and optimal value of the second-stage recourse problem. Regarding computational performance, under four different numbers of sampled day-ahead price trajectories from the Markov process, the proposed method achieves at least $21.87\,\times$ speedup compared to the Gurobi-based method and $7.05\,\times$ speedup over the traditional CCG benchmark. This improved computational efficiency lies in the pre-trained value NN surrogates for the second-stage recourse problem. Under a large number of day-ahead price trajectories, such as in the 1028 synthetic network with 500 day-ahead price trajectories, the NN-accelerated CCG solves the 2S-ARSO within $1{,}232$ seconds, whereas the computational times for Gurobi and the traditional CCG are much longer. As the number of sampled day-ahead price trajectories increases and instances become more challenging, the NN-accelerated CCG delivers lower runtime and improved solution quality. Unlike classical CCG, which alternates between a master problem and a recourse problem, our approach embeds a surrogate of the recourse in a single MILP, thus removing decomposition overhead. Across all test cases, the trade-off between speed and accuracy is favorable: near-optimal solutions are obtained with less computation time, which is valuable in time-critical settings.

\begin{table}[!t]
\centering
\caption{Comparative Results on the 1028-Bus Test System}
\label{tab:1028bus-perf}
\renewcommand{\arraystretch}{1.1}
\setlength{\tabcolsep}{4pt}
\begin{tabular}{llrrrr}
\hline
\multirow{2}{*}{Method} & \multirow{2}{*}{Metric} & \multicolumn{4}{c}{Number of Sampled Price Trajectories} \\
 & & 25 & 100 & 250 & 500 \\
\hline
\multirow{2}{*}{Gurobi} 
 & Obj. (\$) & 1{,}891 & 1{,}670 & 1{,}618 & 1{,}609 \\
 & Time (s)       & 6{,}422 & 23{,}500 & 77{,}294 & 125{,}382 \\
\hline
\multirow{4}{*}{CCG} 
 & Gap (\%)       & 0 & 0 & 0 & 0 \\
 & Time (s)       & 2{,}065 & 7{,}886 & 24{,}230 & 40{,}841 \\
 & \makecell[l]{Speedup vs\\Gurobi}        & 3.11$\times$ & 2.98$\times$ & 3.19$\times$ & 3.07$\times$ \\
\hline
\multirow{6}{*}{Proposed} 
 & Gap (\%)       & 0.001 & 0.001 & 0.001 & 0.001 \\
 & Time (s)       & 293 & 607  &  986.02 &  1232 \\
 & \makecell[l]{Speedup vs\\Gurobi}        & 21.87$\times$ &  38.72$\times$ &  78.39$\times$ &  101.74$\times$ \\
 & \makecell[l]{Speedup vs\\CCG}    & 7.05$\times$  & 12.99$\times$ & 24.57$\times$ & 33.15$\times$ \\
\hline
\end{tabular}
\end{table}

\section{Conclusion}
\label{Conclusion}
This paper introduces an NN–accelerated CCG algorithm for the day-ahead optimal offering problem for DER aggregators. A tailored ReLU-activated NN, encoded through MILP constraints, serves as a surrogate for the recourse problem. The NN enables the master problem to incorporate Benders cuts and scenario selections while preserving feasibility. On large synthetic distribution networks, the approach can achieve 1 to 2 orders of magnitude speedup over traditional CCG benchmarks while maintaining solution optimality. These results demonstrate the effectiveness of the NN-accelerated CCG method in substantially improving the tractability of the network-constrained optimal DER offering problems. Potential future directions include the consideration of topology changes, which can be addressed by augmenting the training set with representative contingencies and encoding contingency-dependent network coefficients as inputs.

\appendix


The following result establishes a finite-termination condition for the NN-accelerated CCG algorithm used in this paper.

\begin{proposition}
\label{prop:finite-termination}
Suppose $\mathcal{X}$ is a polytope with $N$ extreme points, and assume
that the master problem selects an optimal extreme point of
$\mathcal{X}$ at each iteration. Then
Algorithm~\ref{alg:neuraro-ccg-obj} terminates in at most $N+1$
iterations.
\end{proposition}

\begin{proof}
Let $\boldsymbol{x}^{(k)}$ be the optimal solution of the master problem \eqref{NNforMP} at iteration $k$ and let $\mathcal{S}^{(k)} \subseteq \mathcal{U}$ be the active scenario set at iteration $k$. Define $r^{(k)} := \max_{\boldsymbol{\xi} \in \mathcal{S}^{(k)}} \mathrm{NN}_{\Theta^*}\left( \boldsymbol{x}^{(k)},\boldsymbol{\xi} \right)$. Suppose at iteration $k' > k$, the same extreme point $\boldsymbol{x}^{(k)}$ is identified in the master problem. Since $\bar{\boldsymbol{\xi}} \in \argmax_{\boldsymbol{\xi}\in\mathcal{U}}\ \mathrm{NN}_{\Theta^*}(\boldsymbol{x}^{(k)},\boldsymbol{\xi})$ is identified and added to the active scenario set during iteration $k$, we know $\bar{\boldsymbol{\xi}} \in \mathcal{S}^{(k')}$, it then follows $r^{(k')} \ge \argmax_{\boldsymbol{\xi}\in\mathcal{U}}\ \mathrm{NN}_{\Theta^*}(\boldsymbol{x}^{(k')},\boldsymbol{\xi})$. Consequently, the inequality in line \ref{alg:ineq} of Algorithm \ref{alg:neuraro-ccg-obj} does not hold and the algorithm terminates. 
\end{proof}



\bibliographystyle{IEEEtran}
\bibliography{IEEEabrv,references}  

\end{document}